\documentclass{article}
\usepackage{ijcai21}

\usepackage{times}
\usepackage{soul}
\usepackage{url}
\usepackage[hidelinks]{hyperref}
\usepackage[utf8]{inputenc}
\usepackage[small]{caption}
\usepackage{graphicx}
\usepackage{amsmath}
\usepackage{amsthm}
\usepackage{booktabs}
\usepackage[switch]{lineno}
\usepackage{latexsym}
\usepackage{soul}
\usepackage[utf8]{inputenc}
\usepackage{amsmath}
\usepackage{booktabs}
\usepackage{amsthm}
\usepackage{thmtools}
\usepackage{thm-restate}

\usepackage[utf8]{inputenc} 
\usepackage[T1]{fontenc}    
\usepackage{url}            
\usepackage{booktabs}       
\usepackage{amsfonts}       
\usepackage{nicefrac}       
\usepackage{microtype}      

\newtheorem{definition}{Definition}

\usepackage{helvet}  
\usepackage{courier}  
\usepackage{subfig}
\usepackage{enumitem}
\usepackage{multirow}
\usepackage[linesnumbered,ruled,vlined]{algorithm2e}
\usepackage{mathtools}
\usepackage{balance}
\usepackage{flushend}

\newcommand{\cM}{\mathcal{M}}

\newcommand{\cW}{\mathcal{W}}

\newcommand*\system{\textsc{LDP-FL}}

\title{LDP-FL: Practical Private Aggregation in Federated Learning \\ with Local Differential Privacy}

\author{
Lichao Sun$^1$\and
Jianwei Qian$^2$\and
Xun Chen$^2$\thanks{Corresponding Author}\\
\affiliations
$^1$Lehigh University\and
$^2$Samsung Research America
\emails
lis221@lehigh.edu,
qianjavier@gmail.com,
xun.chen@samsung.com
}

\begin{document}

\maketitle

\begin{abstract}
Training deep learning models on sensitive user data has raised increasing privacy concerns in many areas. Federated learning is a popular approach for privacy protection that collects the local gradient information instead of raw data. One way to achieve a strict privacy guarantee is to apply local differential privacy into federated learning. However, previous work does not give a practical solution due to two issues. First, the range difference of weights in different deep learning model layers has not been explicitly considered when applying local differential privacy mechanism. Second, the privacy budget explodes due to the high dimensionality of weights in deep learning models and many query iterations of federated learning. In this paper, we proposed a novel design of a local differential privacy mechanism for federated learning to address the abovementioned issues. It makes the local weights update differentially private by adapting to the varying ranges at different layers of a deep neural network, which introduces a smaller variance of the estimated model weights, especially for deeper models. Moreover, the proposed mechanism bypasses the curse of dimensionality by parameter shuffling aggregation. Empirical evaluations on three commonly used datasets in prior differential privacy work, MNIST, Fashion-MNIST and CIFAR-10, demonstrate that our solution can not only achieve superior deep learning performance but also provide a strong privacy guarantee at the same time. 
\end{abstract}




\section{Introduction}

Many attractive applications involve training models on highly sensitive data, e.g., diagnosis of diseases with medical records, or genetic sequences~\cite{xu2021fedmood,che2021federated}.
In order to protect the privacy of the training data, the federated learning framework is of particular interest since it can provide a well-trained model without touching any sensitive data directly \cite{mcmahan2016communication}.
The original purpose of the federated learning (FL) framework is to share the weights of the model trained on sensitive data instead of data directly.
However, some studies show that the weights also can leak privacy and the original sensitive data can be recovered \cite{papernot2017semi}.
In order to solve the problem, recent works start to use differential privacy to protect the private data in federated learning \cite{nguyen2016collecting,lyu2020privacy,sun2020federated}, but most of them cannot give a practical solution for deep learning on complex datasets due to the trade-off between privacy budget and performance.

Notably, there are at least two remaining challenges in applying local differential privacy (LDP) in FL. First, existing approaches 
have assumed a fixed range of weights for simplicity. However, in complicated models such as those used in deep learning tasks, the range of model weights at different neural network layer varies significantly. Assuming weights in all layers to be in a fixed range will introduce a large variance of the estimated model weights, which leads to the poor model accuracy. 
Second, when client uploads the local model to the server, the server can explore the private connections of model weights, which causes the privacy budget explosion due to the high dimensionality of deep learning models. 

\begin{figure*}[tb]
\centering
\includegraphics[width=5.4in]{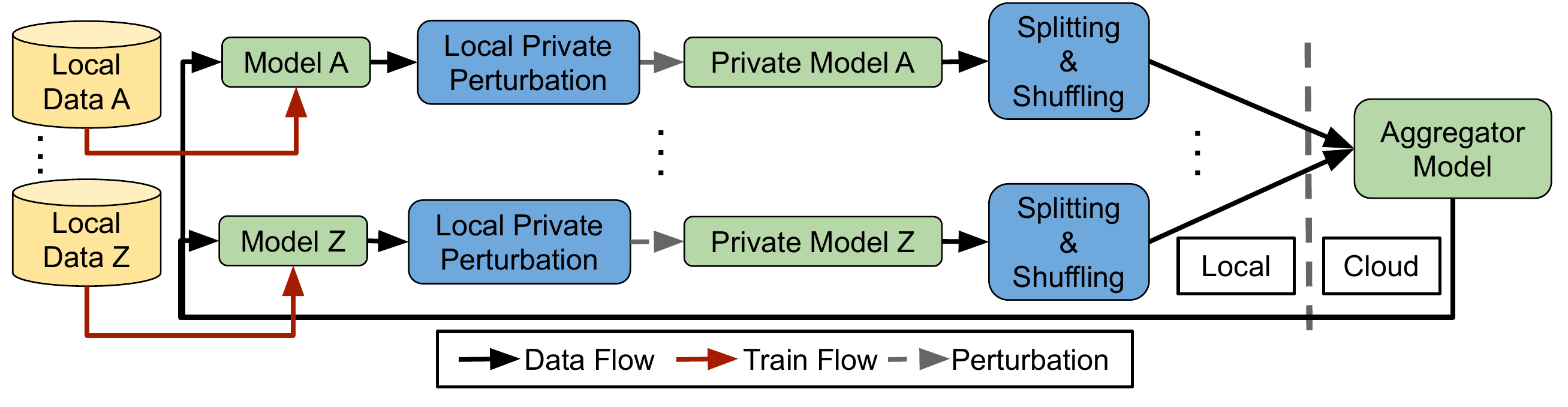}
\caption{The overview of \system. Each client has a private dataset. After training their model locally, clients perturb their model weights with differentially private noise, and locally split and shuffle the weights to make them less likely to be linked together by the cloud. Finally, they send the weights to the cloud for aggregation.} \label{fig:framework}
\vspace{-5pt}
\end{figure*}

In this paper, we proposed Locally Differential Private Federated Learning (\system), a new local differential privacy mechanism to solve the above issues, as shown in Fig. \ref{fig:framework}.
Our main contributions are multifold. First, we propose a new data perturbation with adaptive range by considering that model weights at different deep neural network (DNN) layer could vary significantly. We derive a more general LDP mechanism to illustrate the impact of the range on the variance of the estimated model weights. We also further demonstrate how the proposed adaptive range settings can greatly improve the accuracy of aggregated model, especially in deeper models. To the best of our knowledge, this is the first work that has considered and studied the necessity and effectiveness of adapting to the different model weights' ranges when applying LDP in federated learning. 
Second, we propose a parameter shuffling mechanism to each clients' weights to mitigate the privacy degradation caused by the high data dimensionality of deep learning models and many query iterations. 
Last, we evaluate \system\ on three commonly used datasets in previous works, MNIST \cite{lecun2010mnist}, Fashion-MNIST \cite{xiao2017fashion} and CIFAR-10 \cite{krizhevsky2009learning}. 
The proposed mechanism achieves a privacy budget of $\epsilon=1$ with 0.97\% accuracy loss on MNIST, a privacy budget of $\epsilon=4$ with 1.32\% accuracy loss on FMNIST, and a privacy budget of $\epsilon=10$ with 1.09\% accuracy loss on CIFAR-10, respectively.
Our work gives the first practical LDP solution for federated deep learning with a superior performance.
Due to space limitation, all proofs are in the appendix.

\section{Preliminary}

\noindent\textbf{Federated Learning.}
FL \cite{mcmahan2016communication,konevcny2016federated} has been proposed and widely used in different approaches.
The motivation is to share the model weights instead of the private data for better privacy protection. Each client, the owner of private training data, updates a model locally, and sends all gradients or weights information to the cloud. The cloud aggregates such information from clients, updates a new central model (e.g., averaging all clients' weights), and then distributes it back to a fraction of clients for another round of model update. Such process is continued iteratively until a satisfying performance is achieved. Note that, to minimize communication, each client might take several mini-batch gradient descent steps in local model computation.

\noindent\textbf{Local Differential Privacy}
To enhance the privacy protection, differential privacy (DP) has been applied to federated learning \cite{bhowmick2018protection,geyer2017differentially}. Traditional DP requires a central trusted party which is often not realistic. To remove that limitation, local differential privacy (LDP) has been proposed. 
The definition of $\epsilon$-LDP is given as below:

\begin{definition}~\cite{dwork2011differential}
  A randomized mechanism $\mathcal{M}$ satisfies $\epsilon$-LDP, for any pair input $x$ and $x'$ in $D$, and any output $Y$ of $\mathcal{M}$,
  \begin{equation}
    \Pr [\mathcal{M} (x) = Y] \le {e^\epsilon } \cdot \Pr [\mathcal{M} (x') = Y].
  \end{equation}
\end{definition}
The privacy guarantee of mechanism $\mathcal{M}$ is controlled by privacy budget \cite{dwork2011differential}, denoted as $\epsilon$.
A smaller value of $\epsilon$ indicates a stronger privacy guarantee. 

\section{Overview of \system}

In this section, we introduce \system, a federated learning framework with LDP as shown in Fig.\ref{fig:framework}. It consists of two steps, as described by Algorithm \ref{alg:fedlocal}.

\noindent\textbf{Cloud Update.} 
First, the cloud initializes the weights randomly at the beginning.
Let $n$ be the total number of local clients.
Then, in the $r$-th communication round, the cloud will randomly select $k_r \leq n$ clients to update their weights for local-side optimization.
Unlike the previous works \cite{bhowmick2018protection,erlingsson2014rappor,seif2020wireless}, where they assume that the aggregator already knows the identities (e.g., IP addresses) of the users but not their private data, our approach assumes the client remains anonymous to the cloud. 



\noindent\textbf{Local Update.}
For each client, it contains its own private dataset.
In each communication, the selected local clients will update their local models by the weight from the cloud.
Next, each local model uses Stochastic Gradient Descent (SGD) \cite{RobbMonr51} to optimize the distinct local models' weights in parallel.
In order to provide a practical privacy protection approach,
each local client applies a split and shuffle mechanism on the weights of local model and sends each weight through an anonymous mechanism to the cloud.
For splitting, we separate all parameters of the trained model, and tag them with their location in the network structure.
For shuffling, we randomly generate a time $t$ in Algorithm \ref{alg:splitshuffling} to prevent the cloud from tracking the client owner of each parameter.
Detailed description and discussion on parameter shuffling mechanism are given in the later section.

\subsection{Privacy-Preserving Mechanism} \label{sec:privacyprotection}

In this section, we describe in details the two major components in \system\ for privacy-preservation.





\subsubsection{Data Perturbation with Adaptive Range}

To better understand the necessity of adaptive range in data perturbation, we first generalize the LDP mechanism in~\cite{duchi2018minimax} by considering model weights' range.

Given the weights $W$ of a model, the algorithm returns a perturbed tuple $W^*$
by randomizing each dimension of $W$. Let $\cM$ be our mechanism, for each weight/entry $w\in W$, assuming $w\in [c-r, c+r]$ where $c$ is the center of $w$'s range and $r$ is the radius of the range ($c,r$ depend on how we clip the weight).
Therefore, we design the following LDP mechanism to perturb $w$.

\begin{equation}
\begin{aligned}
\label{eq:ldp}
    w^* &=\cM(w) \\
    &= 
\begin{cases}
    c + r\cdot \frac{e^\epsilon+1}{e^\epsilon-1}, \text{   w.p.   }\frac{(w-c)(e^\epsilon - 1)+r(e^\epsilon+1)}{2r(e^\epsilon+1)}\\ \\
    c - r\cdot \frac{e^\epsilon+1}{e^\epsilon-1},  \text{   w.p.   }\frac{-(w-c)(e^\epsilon - 1)+r(e^\epsilon+1)}{2r(e^\epsilon+1)}
\end{cases}
\end{aligned}
\end{equation}
where $w^*$ is the reported noisy weight by our proposed LDP, and ``w.p.'' stands for ``with probability''. 
Algorithm \ref{alg:noise} shows the pseudo-code of this mechanism. 

From~(\ref{eq:ldp}), we can see that, with a certain privacy budget $\epsilon$, the larger range which indicates the larger $r$, the larger range of perturbed weights $w^*$ which intuitively indicates a noisier perturbed weight. In Section~\ref{sec:privacy}, we give a formal privacy and utility analysis which shows that the larger $r$, the larger variance of the estimated model weights. This observation motivates us to consider the range variation of model weights when applying LDP. Therefore, we propose the adaptive range settings.
The proposed adaptive range setting in private federated learning consists of several steps:

\begin{algorithm}[t]
    \caption{{\system}
	} \label{alg:fedlocal}
    \KwIn{$n$ is the number of local clients; $B$ is the local mini-batch size, $E$ the number of local epochs, $\gamma$ is the learning rate.}
	\textbf{CloudUpdate}\

    Initialize weights $W_0=\{w^0_{id}\mid \forall id\}$, where $id$ indicates the position of each weight\;
    Initialize the range by ($C_0$, $R_0$) for layers of $W_0$\;
    $SendToClient(W_{0}, k_0, C_0, R_0)$\;
	\For{each round $l = 1, 2, \ldots$ \label{line:cfor1}}{
		randomly select $k_l (k_l \leq n)$ local clients\;
		collect all weights updates $\{(id, w_{id})\mid \forall id\}$ from selected clients with $SendToCloud$\;
		\%Now calculate and update model weights $W_l$\;
	    \For{each element $w \in \cW_l$ \label{line:cfor2}}{
	    determine the $id$ for $w$\;
	    $w \leftarrow \frac{1}{k_l} \sum w_{id}$; \% Compute mean of local models\;
	    }
	    \For{each layer of $\cW_l$ \label{line:cfor3}}{
	    update $C$ and $R$ from $W_l$\;
	    }
	    $SendToClient(W_{l}, k_l, C_l, R_l)$ \% Update clients\; 
	}
	
	
	\textbf{LocalUpdate($W_{l}, k_l, C_l, R_l$)}\
	
	Receive weights $W_{l}, k_l, C_l, R_l$ from cloud by $SendToClient$\;
	\For{each local client $s \in k_l$ in parallel \label{line:cfor4}}{
	    $W^s_{l+1} \leftarrow W_l$\;
	    \For{each local epoch $i = 1, 2, ...E$ \label{line:cfor5}}{
	        \For{each batch $b \in B$ \label{line:cfor6}}{
                $W^s_{l+1} \leftarrow W^s_{l+1} - \gamma \bigtriangledown L(W^s_{l+1}; b)$\;
            }
        }
        DataPertubation($W^s_{l+1}, C_{l}, R_{l}$)\;
        ParameterShuffling($W^s_{l+1})$\;
	}
\end{algorithm}
\begin{algorithm}[] 
    \caption{DataPertubation	
	} \label{alg:noise}
	\KwIn{Original local weights $W^s_{l+1}$, range represented by $C_l$ and $R_l$, privacy budget $\epsilon$}
	\KwOut{Perturbed weights $W^s_{l+1}$}
	
	\For{each $w \in W^s_{l+1}$ and corresponding $c \in C_{l}$ and $r \in R_{l}$ \label{line:cfor7}} {
        
		Sample a Bernoulli variable $u$ such that\\
		$\Pr[u=1]=  \frac{(w-c)(e^\epsilon - 1)+r(e^\epsilon+1)}{2r(e^\epsilon+1)}$\;
		\If{u==1}{
		    $w^* = c + r\cdot \frac{e^\epsilon+1}{e^\epsilon-1}$\;
		}
		\Else{
		    $w^* = c - r\cdot \frac{e^\epsilon+1}{e^\epsilon-1}$\;
		}
	    $w \leftarrow w^*$; \% update $W^s_{l+1}$\; 
	    }
	\bf{return}.
\end{algorithm}
\begin{algorithm}[t]
    \caption{ParameterShuffling	
	} \label{alg:splitshuffling}
	\KwIn{Perturbed weights $W^s_{l+1}$ 
	after Algorithm \ref{alg:noise}}
    label the position $id$ of each element of $W$\;
	\For{each element $w^s \in W$} {
	    label the element position with a unique $id$
	    $t^s_{id} \leftarrow U(0, T)$ \% Randomly sample a small latency between 0 and $T$;
	}
    $SendToCloud(id, w_{id})$ at time $t^s_{id}$\;
	\bf{return}.
\end{algorithm}

\begin{figure*}[tb]
\centering
\includegraphics[width=5in]{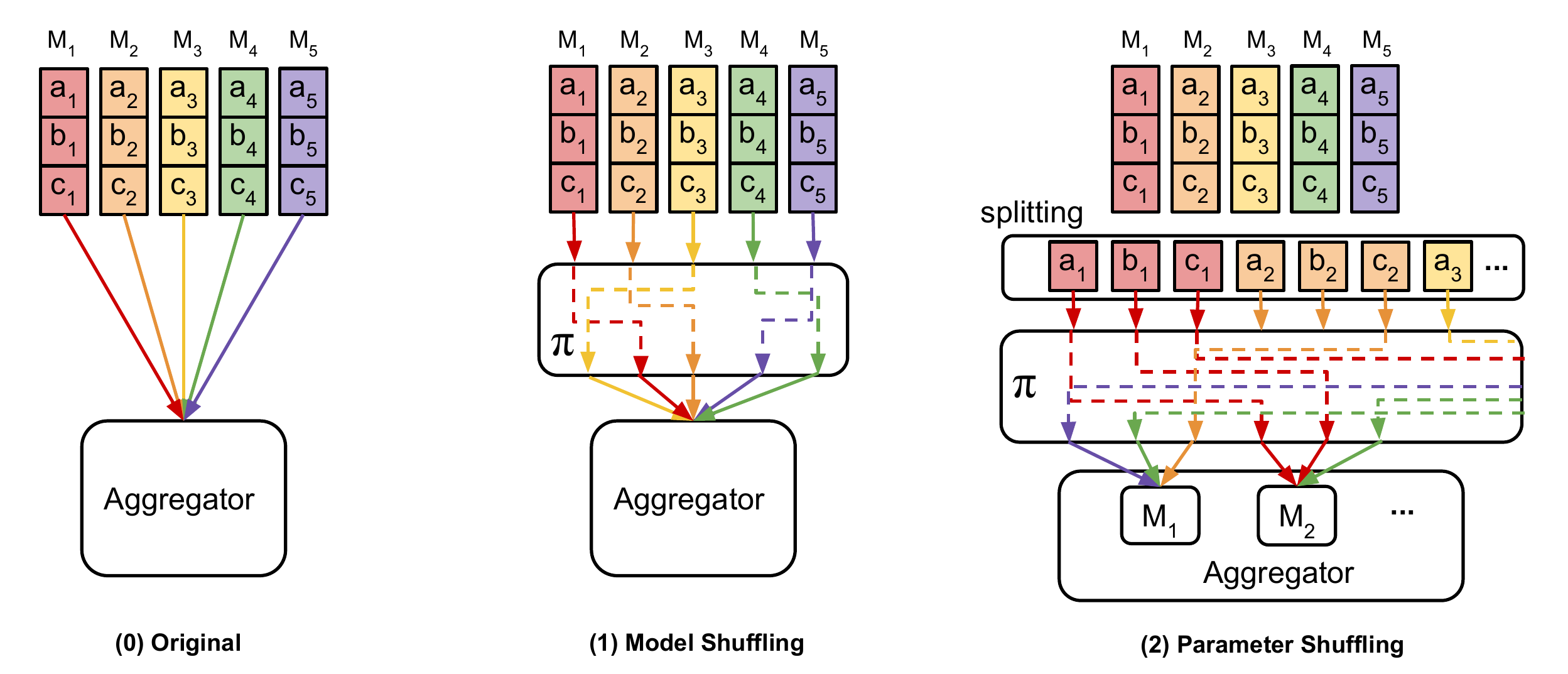}
\caption{Example of parameter shuffling in \system. $\pi$ is a random shuffling mechanism. Note that, the parameter shuffling is composed of parameter shuffling mechanisms. The cloud will aggregate all kinds of information from clients.} \label{fig:shuffling}
\vspace{-5pt}
\end{figure*}
\begin{itemize}
    \item All clients and the server agree to the same weight range, represented by $(C_0, R_0)$ based on prior knowledge at the beginning when initializing weights.
    \item In local update, based on the weight range $(c_l, r_l)$ of each layer where $c_l \in C$ and $r_l \in R$, each local client can perturb its weights by the proposed LDP data perturbation given by~(\ref{eq:ldp}).
    \item In cloud update, the cloud calculates and updates $C$ and $R$ from the local weights update received from clients. The updated $C$ and $R$ will be distributed to the clients when the cloud distributes the updated weights.
\end{itemize}
Note that range parameters $C$ and $R$ are calculated layer-wise, i.e., they are vectors with each element corresponding to one layer in the model. A pseudo-code of adaptive range settings is given in Algorithm~\ref{alg:fedlocal}.
Our adaptive range setting gives a new feasible approach to set up the weight range and improves the utility for complex deep learning tasks.

%
%
%
%
\subsubsection{Parameter Shuffling}
As noted in \cite{dwork2014algorithmic}, a sophisticated privacy-preserving system composites multiple locally differentiate algorithms that lead to a composition of privacy cost of such algorithms. That is, a combination of locally differentiate algorithm with privacy budget $\epsilon_1$ and $\epsilon_2$ consumes a privacy budget of $\epsilon_1 +\epsilon_2$. 
Training a DNN using federated learning requires clients to upload gradient updates to the cloud in multiple iterations. If LDP is applied at each iteration along the iterative training process, the privacy budget will accumulate, which causes the explosion of the total privacy budget. 

Previous efforts \cite{bittau2017,erlingsson2019amplification} have shown that when data reports are properly anonymized at each timestep and unlinkable over the time, the overall privacy guarantee can be greatly improved. That is, by adding anonymity onto the local model updates in \system, it can break the linkage from the received data at the cloud to a specific client and decouple the gradient updates sent from the same client in each iteration.
This can be achieved by multiple existing mechanisms, depending on how the server tracks clients in the specific scenario. As a typically best practice for privacy, a certain level of anonymity of each client to the server is applied to disassociate clients' individually-identifiable information from their weight updates. For example, if the server tracks clients by IP address, each client can adopt an untraceable IP address by using a network proxy, a VPN service \cite{OperaVPN}, public WiFi access, and/or the Tor \cite{Dingledine04tor}. For another example, if the server tracks clients by software-generated metadata such as an ID,  each client can randomize this metadata before sending it to the server.

However, we argue that client anonymity is not sufficient to prevent side-channel linkage attacks, e.g., the cloud can still link a large number of weight updates together if the client uploaded them at the same time in each iteration.
Therefore, we design the local parameter shuffling mechanism to break the linkage among the model weight updates from the same clients and to mix them among updates from other clients, making it harder for the cloud to combine more than one piece of updates to infer more information about any client.

Fig.~\ref{fig:shuffling} illustrates the idea of parameter shuffling of the weights of each local client model. For local model $M_1, M_2, M_3, M_4, M_5$, each model has the same structure but with different weight values. In comparison, original federated learning sends the models' information to the cloud, as shown in Fig.~\ref{fig:shuffling} (0), and \cite{balle2019privacy,cheu2019distributed,erlingsson2019amplification} uses the model shuffling to anonymize the communications between cloud and client, as shown in Fig.~\ref{fig:shuffling} (1). However, prior works doesn't consider the privacy issue due to the high dimensionality of the DNN weights. 
To solve the curse of the high dimensionality of DNN models, parameter shuffling is executed in two steps. In the first step, each client splits the weights of their local model, but labels each weight with an id to indicate its location of the weight in the network structure. In the second step, each client samples a small random latency $t$ from a uniform distribution $U(0, T)$, where $T>0$, for each weight and waits for $t$ before sending the weight to the cloud. Note that $T$ is agreed by all clients at the beginning of federated learning, e.g., by broadcasting their proposed $T$ and computing the median of them (not computing average because average is vulnerable to a single dishonest client's influence). Since all uploads happen uniformly randomly during the same period of time, the cloud cannot distinguish them by upload time and cannot associate the weights updates from the same client. An approach is given in Algorithm~\ref{alg:splitshuffling}.

%
In practical scenarios, the clients might spend different time in training (due to different hardwares) and in communications (due to different network conditions). The above design of the random delay for each client can accommodate such a heterogeneity with a minor adjustment.
For client $i$, we denote the local computing time, the communication time, the proposed random delay for $j$-th model parameter by $t_{LC}^i$, $t^i_C$ and $t_j^i$, respectively. Define $T_S = \max_i\{t_{LC}^i + t_{C}^i\}$, which represents the response time of the slowest client without parameter shuffling. At the beginning of each iteration, $T_S$ can be estimated from clients' predictions of the corresponding $t_{LC}^i$ and $t_{C}^i$ based on their hardware specifications and communication settings.
Instead of having $t^i_{j} \sim U(0,T)$ as originally proposed in the paper, we set $t^i_{j} = T_W^i + \tilde{t}^i_j$, where $T_W^i = T_S - t_{LC}^i - t_{C}^i$ and $\tilde{t}^i_j\sim U(0,T)$. 
That is, when sending any parameter, client $i$ waits $T_W^i$ before adding a random delay uniformly distributed from 0 to $T$.
The resulting response time for $j$-th model parameter by client $i$, defined as $t_{LC}^i + t^i_{j} + t_{C}^i$, can be then rewritten as $T_S + \tilde{t}^i_j$. 
This shows that adding $T_W^i$ has no negative effect and the additional delay overhead is still controlled by $T$, as the delay of each iteration is decided by the slowest client (at least $T_S$). The response times of all parameter updates are still random and uniformly distributed. Therefore, the client anonymity is preserved, and the privacy budget won't accumulate. Besides, additional methods, such as not selecting the slower device (as the reviewer also mentioned), optimizing the local computing or applying asynchronous FL, can be applied regardlessly.

In theory, $T$ can be as small as clients and cloud can support, to add enough randomness to the parameter shuffling. In practice, $T$ might be given a safe margin by considering the imperfectness of time synchronization and the estimation of $T_S$, and the unavoidable randomness in local computing (for example, device sluggish) and communications (for example, packet drop), so that the parameter updates reaches the cloud randomly. In any case, $T$ is expected to be a small impact on the time cost of each FL iteration. 

\section{Privacy and Utility Analysis} \label{sec:privacy}
In this section, we analyze the privacy guarantee and utility cost of our approach, and compare it with other LDP mechanisms. Due to limited space, all proofs of lemmas and theorems are given in the appendix.

\noindent\textbf{Local Differential Privacy}
We prove the proposed random noise mechanism $\cM$ satisfies LDP w.r.t. a bounding range. 

\begin{restatable}{theorem}{epsproof}
Given any single number $w \in [c-r, c+r]$, where $c$ is the center of $w$'s range and $r$ is the radius of the range, the proposed mechanism $\cM$ in Eq.~\ref{eq:ldp} satisfies $\epsilon$-LDP w.r.t. $[c-r, c+r]$.
\end{restatable}
Both $\epsilon$, $r$ affect the privacy level. $\epsilon$ determines how well the data is hidden in a ``crowd'' and $r$ determines the size of the ``crowd''.
Suppose the true average weights in an iteration is $\Bar{w}=\dfrac{1}{n}\sum_u w_u$ for each $w \in W$, the proposed LDP mechanism in Eq.~\ref{eq:ldp} induces zero bias in calculating the average weight $\overline{\mathcal{M}(w)}=\dfrac{1}{n}\sum_u \mathcal{M}(w_u)$.
\begin{restatable}{lemma}{noiseproof}
\label{th:bias}
Algorithm~\ref{alg:noise} introduces zero bias to estimating average weights, i.e., $\mathbb{E}[\overline{\mathcal{M}(w)}]=\Bar{w}$.
\end{restatable}

The following lemmas show the variance of the data obtained by the cloud.
First, the proposed mechanism leads to a variance to each reported weight $\mathcal{M}(w)$. 
\begin{restatable}{lemma}{varcmpproof}
\label{lemma:cmp}
Let $\cM$ be the proposed data perturbation mechanism. Given any number $w$, the variance of the mechanism is
$Var[\mathcal{M}(w)] =r^2\left(\dfrac{e^\epsilon+1}{e^\epsilon-1}\right)^2.$
\end{restatable}
Then, for the variance of the estimated average weight $\overline{\mathcal{M}(w)}$,
we have the lower bound and upper bound as below.
\begin{restatable}{lemma}{varianceproof}
\label{lemma:var}
Let $\overline{\mathcal{M}(w)}$ be the estimated average weight. The lower bound and upper bound of the estimated average weight is:
$\dfrac{r^2(e^\epsilon+1)^2}{n(e^\epsilon-1)^2}-\dfrac{r^2}{n} \le Var[\overline{\mathcal{M}(w)}] \le \dfrac{r^2(e^\epsilon+1)^2}{n(e^\epsilon-1)^2}.$
\end{restatable}
It can be seen from Lemma \ref{lemma:var} that, with enough clients (a large $n$), we can achieve a low variance on the estimated average weight with a wide range of perturbed data (a large $r$) and a small privacy cost (a small $\epsilon$). 
With the above lemmas, we conclude the following theorem that shows the accuracy guarantee for calculating average weights on the cloud.
\begin{restatable}{theorem}{boundproof}
For any weight $w\in W$, with at least $1-\beta$ probability, $|\overline{\mathcal{M}(w)}-\Bar{w}|<O\left(\dfrac{r\sqrt{-\log\beta}}{\epsilon\sqrt{n}}\right).$
\end{restatable}

\noindent\textbf{Adaptive Range Setting}
If the weights' ranges are adaptive, clients may choose different $c,r$ to bound their weights and a client may choose different $c,r$ for each of its weights in each round. In such cases, LDP is still satisfied because the noise is applied to each weight separately. It is straightforward that the bias of the estimated average weight is still 0.
For variance, we can derive a corollary from Lemma~\ref{lemma:var}.
\begin{restatable}{corollary}{variancecorollary}
When client $u$ chooses a radius $r_u$, the variance bound of the estimated average weight is:
$\dfrac{(e^\epsilon+1)^2\sum_ur_u^2}{(e^\epsilon-1)^2n^2}-\dfrac{\sum_ur_u^2}{n^2} \le Var[\overline{\mathcal{M}(w)}] \le \dfrac{(e^\epsilon+1)^2\sum_ur_u^2}{(e^\epsilon-1)^2n^2}.$
\end{restatable}




\noindent\textbf{Parameter Shuffling}
As we claimed before, \system\ uses parameter shuffling to bypass the curse of dimensionality.
Since the weights are split and uploaded anonymously, the cloud is unable to link different weight values from the same client, so it cannot infer more information about a particular client. Therefore, it is sufficient to protect $\epsilon$-LDP for each weight. Likewise, because of the client anonymity, the cloud is unable to link weights from the same client at different iterations.
Without parameter shuffling, the privacy budget of LDP will grow to $Td\epsilon$, where $T$ is the iteration number and $d$ is the number of weights in the model.
Model shuffling \cite{balle2019privacy,cheu2019distributed,erlingsson2019amplification} can deduce the privacy budget to $d\epsilon$, which still cursed by the high dimension of DNN weight.
Similar discussion can be found in~\cite{erlingsson2014rappor,erlingsson2019amplification}. Unlike our approach, the previous works shuffle the ordered sequential information.

\begin{figure*}[tb]
\centering
\subfloat[Effect of $n$ on MNIST]{\includegraphics[width=1.8in]{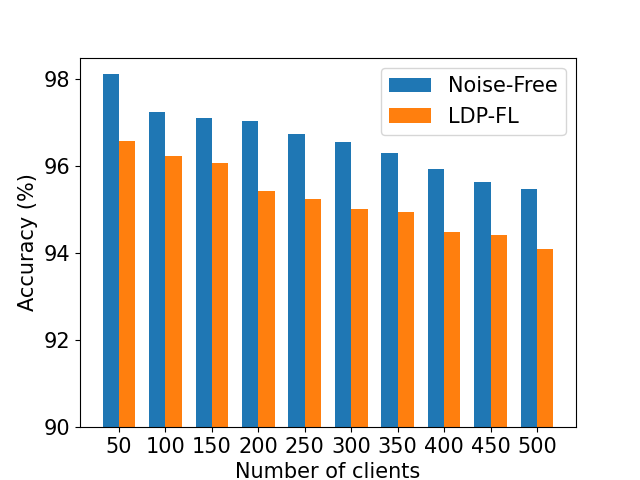}}\ \ \
\subfloat[Effect of $n$ on FMNIST]{\includegraphics[width=1.8in]{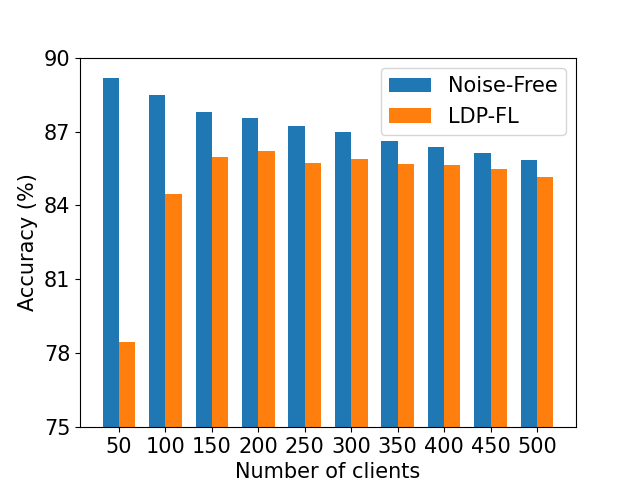}}\ \ \
\subfloat[Effect of $n$ on CIFAR-10]{\includegraphics[width=1.8in]{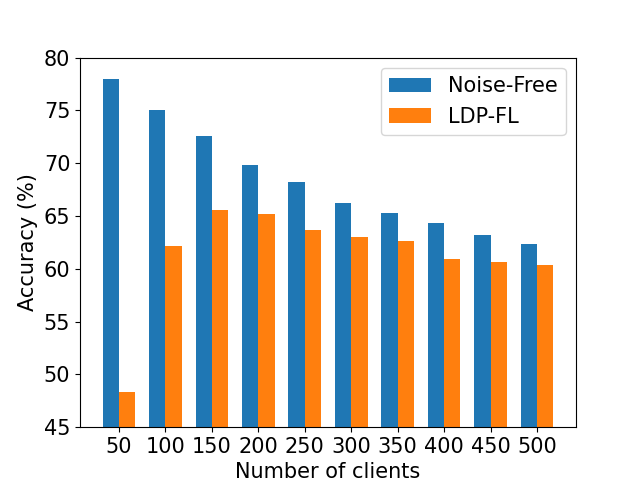}}\\
\caption{Effect of $n$ on the training accuracy.} \label{fig:budget}
\vspace{-5pt}
\end{figure*}

\vspace{-5pt}
\section{Experiments}


In this section, 
we examine the effect of different weights based on the image benchmark datasets, MNIST \cite{LeCun1998}, Fashion-MNIST (FMNIST) \cite{xiao2017fashion} and then verify the performance improvement based on CIFAR-10 \cite{krizhevsky2009learning} and the preceding three datasets. 
For MNIST and FMNIST, we implement a two-layer CNN for image classification. However, for CIFAR-10, the default network from Pytorch library only achieves around 50\% accuracy without data perturbation, and we re-design a small VGG \cite{simonyan2014very} for the task.
The training data and the testing data are fed into the network directly in each client, and for each client, the size of the training data is the total number of the training samples divided by the number of the clients.
In this case, a larger number of clients indicates the small size of the training data of each client.
For each weight, we clip them in a fixed range. In this work, we set $(c, r)=(0, 0.075)$ and $(0, 0.015)$ by default for MNIST and FMNIST, respectively. However, for CIFAR-10, instead of using a fixed range, due to the complexity of the model, we set $c$ and $r$ adaptively by the weight range of each layer.
The learning rate $\gamma$ is set as 0.03 for MNIST/FMNIST and 0.015 for CIFAR-10.
Considering the randomness during perturbation, we run the test experiments ten times independently to obtain an averaged value.
In order to evaluate the performance of different learning approaches, we use different metrics including accuracy for utility, $\epsilon$ for privacy cost and $m$, the number of communication rounds (CRs), for communication cost. Any approach with a high accuracy, a small $\epsilon$ and a small communication round indicates a good and practical solution.
The proposed models are implemented using Pytorch, and all experiments are done with a single GPU NVIDIA Tesla V100 on the local server. 
Experiments on MNIST and FMNIST can be finished within an hour with 10 CRs, and experiments on CIFAR-10 need about 2 hours with 15 CRs.




\subsection{Parameter Analysis}
\begin{figure}[t]
\centering
\subfloat[Communication Round $m$]{\includegraphics[width=1.65in]{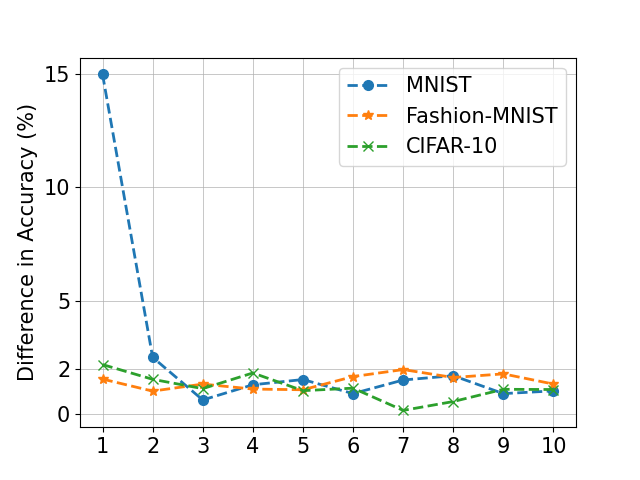}}
\subfloat[Fraction of Clients $f_r$]{\includegraphics[width=1.65in]{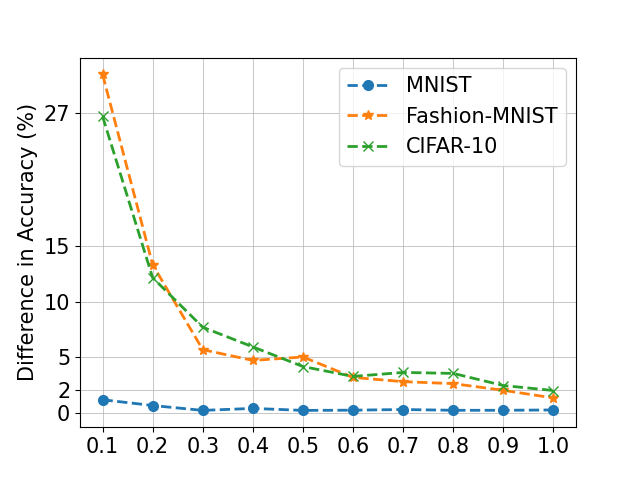}}
\caption{Parameter analysis on communication round and fraction of clients. For MNIST, FMNIST and CIFAR-10, we set the number of total clients as 100, 200, 500 respectively.} \label{fig:param}
\vspace{-5pt}
\end{figure}

\begin{figure*}[tb]
\centering
\subfloat[Effect of $\epsilon$ on MNIST]{\includegraphics[width=1.65in]{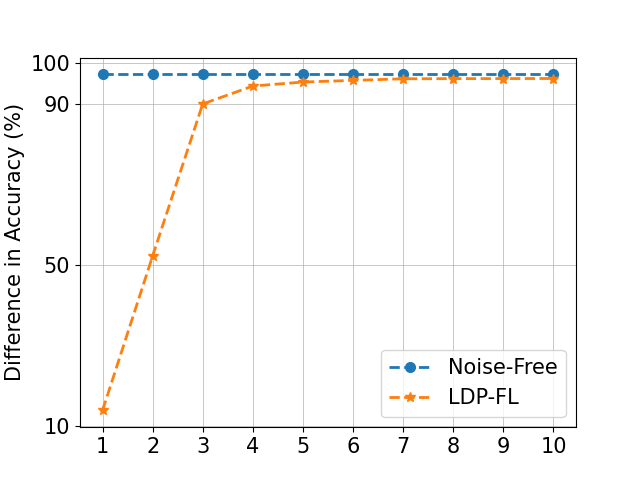}} \ \ \
\subfloat[Effect of $\epsilon$ on FMNIST]{\includegraphics[width=1.65in]{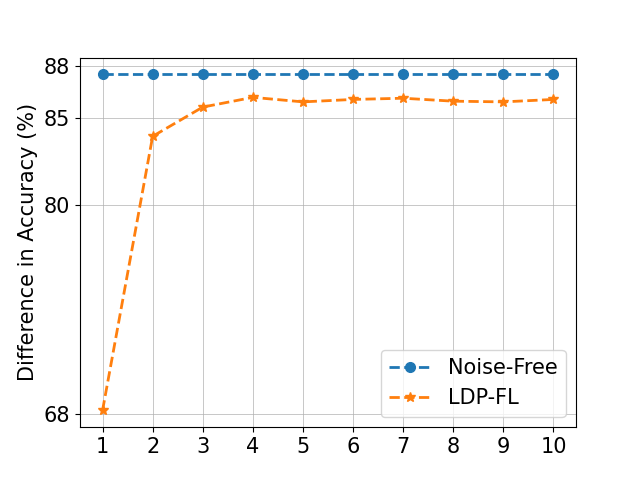}} \ \ \
\subfloat[Effect of $\epsilon$ on CIFAR-10]{\includegraphics[width=1.65in]{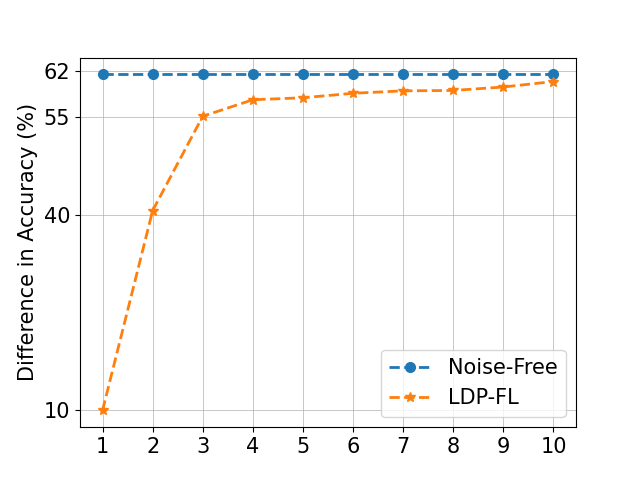}} \ \ \
\subfloat[Effect of $r$ on CIFAR-10]{\includegraphics[width=1.7in]{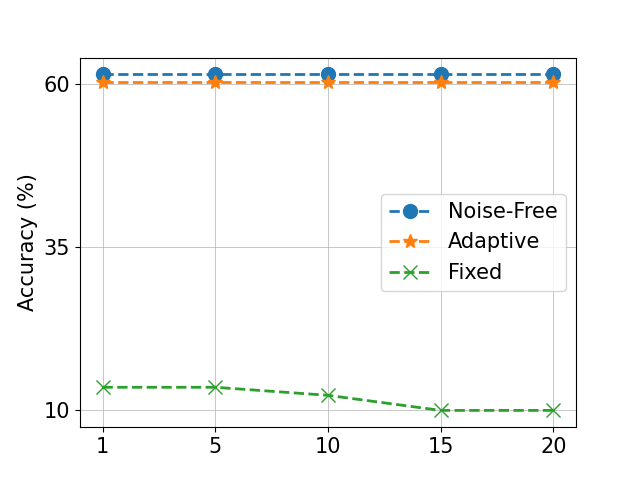}}
\caption{[a-c]: Effect of $\epsilon$ on training accuracy; [d]: Effect of $r$ on the training accuracy} \label{fig:p1}
\vspace{-5pt}
\end{figure*}

Fig.~\ref{fig:param} shows the change of the performance with varying $f_r$ and $m$. In this work, we evaluate the proposed model on MNIST, FMNIST and CIFAR.
To evaluate the parameter $f_r$, we first fix the number of the client as 500. It can be found that when $f_r$ is too small, it does not affect the performance on MNIST, but affects significantly the performance on FMNIST and CIFAR-10.
When $f_r$ is close to 1, \system\ can achieve almost the same performance as the noise-free results on MNIST, FMNIST and CIFAR-10, as shown in Fig. \ref{fig:p1}(a-c).
Another important parameter is the communication rounds between cloud and clients. It is not hard to see that with more communications, we can train a better model on all datasets through the proposed model. However, due to the complexity of data and tasks, CIFAR-10 needs more communication rounds for a better model.

Fig. \ref{fig:budget} shows that \system\ can achieve a high accuracy with a low privacy cost because of the new design of the communication and the new local data perturbation. While increasing the number of clients $n$ in training, the \system\ can perform as close as the noise-free federated learning. Compared with MNIST ($n=100$, $\epsilon=1$), FMNIST ($n=200$, $\epsilon=5$), CIFAR-10 ($n=500$, $\epsilon=10$) needs more clients, which indicates that for a more complex task with a larger neural network model, it requires more local data and more clients for a better performance against the data perturbation. The privacy budget also affects the performance of the central model, but more details will be discussed in the next section.

\subsection{Performance Analysis and Comparison}
In Fig.\ref{fig:p1}(a-c), \system\ can achieve 96.24\% accuracy with $\epsilon=1$ and $n=100$, 86.26\% accuracy with $\epsilon=4, n=200$ and 61.46\% accuracy with $\epsilon=10, n=500$ on MNIST, FMNIST and CIFAR-10, respectively. 
Our results are very competitive comparing to other previous works.
\cite{geyer2017differentially} first apply DP on federated learning. While they use $100$ clients, they can only achieve 78\%, 92\% and 96\% accuracy with $(\epsilon, m) = (8, 11)$, $(8, 54)$ and $(8, 412)$ guaranteed on MNIST with differential privacy, where $(\epsilon, m)$ represents the privacy budget and the communication rounds.
\cite{bhowmick2018protection} first utilize the LDP in federated learning. Due to the high variance of their mechanism,  it requires more than 200 communication rounds and spends much more privacy budgets, \textit{i.e.,} MNIST ($\epsilon=500$) and CIFAR-10 ($\epsilon=5000$). 
Last, the most recent work \cite{truex2020ldp} utilizes Condensed Local Differential Privacy ($\alpha$-CLDP) into federated learning, with 86.93\% accuracy on FMNIST dataset. However, $\alpha$-CLDP achieved that performance by requiring a relatively large privacy budget $\epsilon = \alpha \cdot 2c \cdot10^\rho$ (e.g. $\alpha=1, c=1, \rho=10$), which results in a weak privacy guarantee. Compared to existing works, our approach needs many fewer communication rounds between clients and cloud (\textit{i.e.,} 10 for MNIST, 15 for FMNIST and CIFAR-10), which makes the whole solution more practical in real scenarios.
Overall, \system\ achieves better performance on both effectiveness, efficiency and privacy cost than prior works.


\subsection{Analysis of Privacy Budget}

    
To analyze the impact of privacy budgets on performance, we present the accuracy results with $\epsilon$ from 0.1 to 1 for MNIST, 1 to 10 FMNIST and 1 to 10 for CIFAR-10 in Fig. \ref{fig:p1} (a-c). 
It shows that \system\ can maintain the accuracy at a high value for a wide range of privacy budgets on MNIST ($\epsilon > 0.3$), FMNIST ($\epsilon > 2$) and CIFAR-10 ($\epsilon > 5$).
While $\epsilon<0.3$ for MNIST, $\epsilon < 2$ for FMNIST and $\epsilon < 3$ for CIFAR-10, the accuracy drops significantly.
It is obvious that more complex data and tasks require more privacy cost.
The main reason is that the complex task requires a sophisticated neural network, which contains a large number of model weights.
Meanwhile, the range of each weight is also wider in the complex task, which causes a larger variance for the estimated weights.


\subsection{Analysis of Adaptive Range Setting}

\begin{table}[tb]
    \centering
  
  \resizebox{2.450in}{!}{%
  \begin{tabular}{|l|c|c|c|}
\hline
                   & $\epsilon$   & Fixed & Adaptive  \\ \hline
MNIST  (n=100) & 5 & 96.34\% & 96.06\%       \\ \hline
MNIST  (n=100) & 10 & 96.39\% & 96.11\%       \\ \hline
FMNIST  (n=200) & 5 & 85.95\% & 87.54\%       \\ \hline
FMNIST  (n=200) & 10 & 86.10\% & 87.57\%       \\ \hline
CIFAR (n=500)  & 5 & 13.54\%       & 58.89\% \\ \hline
CIFAR (n=500)  & 10 & 14.28\%       & 60.37\%  \\ \hline
\end{tabular}
  }
  \caption{Comparison of Adaptive and Fixed Weight Range Settings} \label{table:adaptive}
  \vspace{-10pt}
\end{table}

To analyze the effectiveness of the adaptive weight range, we compare the accuracy with fixed range versus that with adaptive range for all three benchmark datasets in Table \ref{table:adaptive}. We run the experiments for the privacy budget of $\epsilon = 5$, and $\epsilon = 10$, respectively. Clearly, for CIFAR-10, the adaptive range setting presents the greatest performance improvement over the fixed range setting. For example, in CIFAR-10 with $\epsilon =10$, if $(c,r)=(0,1)$ is chosen for the fixed ranging setting, the adaptive range setting improves the accuracy by 45\% over the fixed range setting.
While on less complex datasets such as FMNIST and MNIST, we observe that the advantage of adaptive range setting gradually disappears. Specifically, on the MNIST dataset, there is no obvious difference in the accuracy achieved by the adaptive range setting and by the fixed range. This is because the model on CIFAR-10 is much deeper and more complex, such that the model weights' range in each layer differs significantly. In such a case, the adaptive range setting is able to reduce the variance of the estimated weights in some layers when the corresponding weight range is small. This improves the model accuracy dramatically. The results also explain why previous approach using fixed range can achieve a good accuracy on simple and shallow model. 

We take a closer look at the performance for CIFAR-10 with $\epsilon =10$. As illustrated in Figure~\ref{fig:p1}(d), no matter how we tune the value of $r$ for fixed range setting, the model accuracy always stay around 20\% and can never be improved. The reason might be that the larger $r$ increases the variance of model weights estimation, while the smaller $r$ may generate the perturbed weights irrelevant to the true data range and cause an accuracy loss. We argue that for deeper neural networks, the fixed range setting might completely fail.







\vspace{-5pt}
\section{Related Work}


Differential privacy \cite{dwork2014algorithmic} provides a mathematically provable framework to design and evaluate a privacy protection scheme. Traditional differential privacy \cite{dwork2014algorithmic}, also known as Central Differential Privacy (CDP), considers the private data publication of the aggregated raw data. It usually assumes that a trusted data curator in the cloud perturbs the data statistic information by differential privacy mechanism. However, such assumption is often unrealistic in many use cases. Therefore, a new differential privacy model, Local Differential Privacy (LDP), was recently proposed \cite{duchi2013local}, in which, the raw data is perturbed before being sent to the data curator in the cloud.

Recently, more works studies how to use DP in federated learning \cite{geyer2017differentially,mcmahan2017learning} or use LDP in the federated learning \cite{bhowmick2018protection,nguyen2016collecting,seif2020wireless,truex2020ldp}.
However, existing approaches can not be applied practically on deep learning. 
Some of them \cite{nguyen2016collecting,seif2020wireless} only focus on small tasks and simple datasets and do not support deep learning approaches yet.
The other works \cite{bhowmick2018protection,truex2020ldp} studied the LDP on federated learning problem. However, as we discussed in the experiments part, both of them hardly achieve a good performance with a small limited budget.
Bhowmick et al. \cite{bhowmick2018protection} is the first work focus on the LDP on federated learning with deep learning models. However, it first need a strong assumption that attacker has little prior acknowledge about user data. Meanwhile, it spends $\epsilon = 500$ for MINST and $\epsilon = 5000$ for CIFAR-10, which means the LDP is not practical and needs traditional differential privacy for second protection.
Last, they did not report the absolute performance on MNIST and CIFAR-10. 
Truex et al. \cite{truex2020ldp} is another recent work. It uses It contains two obvious problems. First, it replaces LDP by Condensed LDP (CLDP), which spends a large privacy budget $\epsilon$ with a large clip bound reported in the work. Second, this work does not provide a complete privacy utility analysis of the number of weights and communication rounds.
Due to the high variance of the noises in their approaches, it requires more communication rounds between cloud and clients, which spends more privacy cost in their framework.
Our approach focus on solving the weakness of all previous approaches and accelerates the practical solution on complex dataset.


\vspace{-5pt}
\section{Conclusion}
\vspace{-2pt}

In this paper, we propose a new LDP mechanism for federated learning with DNNs. By giving a generalized data perturbation that is applicable for arbitrary input ranges, we proposed an adaptive range setting for improving privacy/utility trade-off. We also designed a parameter shuffling mechanism to mitigate the privacy degradation caused by high data dimension and many query iterations.
Our empirical studies demonstrate that the performance of \system\ is superior to the previous related works on the same image classification tasks. Hopefully, our work can considerably accelerate the practical applications of LDP in federated learning.


\newpage

\bibliography{reference}
\bibliographystyle{named}

\newpage

\appendix

\section{Proof}
\epsproof*
\begin{proof}
We know the weight $w$'s range is $[c-r, c+r]$.
If $w^*=c+r\cdot\frac{e^\epsilon+1}{e^\epsilon-1}$, then for any $w,w'\in[c-r, c+r]$.
\begin{equation}
\begin{aligned}
\dfrac{\Pr[\mathcal{M}(w)=w^*]}{\Pr[\mathcal{M}(w')=w^*]}&\le \dfrac{\underset{w}{\max}\Pr[\mathcal{M}(w)=w^*]}{\underset{w'}{\min}\Pr[\mathcal{M}(w')=w^*]} \\
&= \dfrac{r(e^\epsilon-1)+r(e^\epsilon+1)}{-r(e^\epsilon-1)+r(e^\epsilon+1)} = e^\epsilon.
\end{aligned}
\end{equation}
If $w^*=c-r\cdot\frac{e^\epsilon+1}{e^\epsilon-1}$, the above still holds.
\end{proof}

\noiseproof*
\begin{proof}
For any weight update $w_u$ from any client $u$,
\begin{equation}
\begin{aligned}
\mathbb{E}[\mathcal{M}(w_u)]&=\left(c+r\cdot\dfrac{e^\epsilon+1}{e^\epsilon-1}\right)\cdot \dfrac{(w_u-c)(e^\epsilon-1)+r(e^\epsilon+1)}{2r(e^\epsilon+1)} \\
 + & \left(c-r\cdot\dfrac{e^\epsilon+1}{e^\epsilon-1}\right)\cdot \dfrac{-(w_u-c)(e^\epsilon-1)+r(e^\epsilon+1)}{2r(e^\epsilon+1)}\\
= & \dfrac{2r(w_u-c)(e^\epsilon-1)}{2r(e^\epsilon-1)}+c \\
= & w_u.
\end{aligned}
\end{equation}
\begin{equation}
\begin{aligned}
\mathbb{E}[\overline{\mathcal{M}(w)}]&=\mathbb{E}\left[\dfrac{1}{n}\sum_u \mathcal{M}(w_u)\right] = \dfrac{1}{n}\sum_u \mathbb{E}[\mathcal{M}(w_u)] \\
&= \dfrac{1}{n}\sum_u w = \Bar{w}.
\end{aligned}
\end{equation}
\end{proof}

\varcmpproof*
\begin{proof}
The variance of each reported noisy weight $w_u$ is
\begin{align}
\text{Var}[\mathcal{M}(w_u)] &= \mathbb{E}(\cM^2(w_u)) - \mathbb{E}^2(\mathcal{M}(w_u)) \nonumber \\ 
&= r^2\left(\dfrac{e^\epsilon+1}{e^\epsilon-1}\right)^2-(w_u-c)^2 \\
&\le r^2\left(\dfrac{e^\epsilon+1}{e^\epsilon-1}\right)^2. \nonumber
\end{align}
\end{proof}

\varianceproof*
\begin{proof}
The variance of the estimated average weight is
\begin{align}
\text{Var}[\overline{\mathcal{M}(w)}] &= \dfrac{1}{n^2}\sum_u\text{Var}[\mathcal{M}(w_u)] \nonumber \\
&= \dfrac{r^2(e^\epsilon+1)^2}{n(e^\epsilon-1)^2} -\dfrac{1}{n^2}\sum_u(w-c)^2.
\end{align}
\text{So the bound of variance is }
\begin{equation}
\label{eq:var}
 \dfrac{r^2(e^\epsilon+1)^2}{n(e^\epsilon-1)^2}-\dfrac{r^2}{n} \le \text{Var}[\overline{\mathcal{M}(w)}] \le \dfrac{r^2(e^\epsilon+1)^2}{n(e^\epsilon-1)^2}.
\end{equation}
\end{proof}

\boundproof*
\begin{proof}
For each client $u$, $|\mathcal{M}(w_u)-w_u|\le r\cdot\dfrac{e^\epsilon+1}{e^\epsilon-1}+r = \dfrac{2re^\epsilon}{e^\epsilon-1}$ and
$\text{Var}[\mathcal{M}(w_u)]=\text{Var}[\mathcal{M}(w_u)-w_u]=\mathbb{E}[(\mathcal{M}(w_u)-w_u)^2]-\mathbb{E}^2[\mathcal{M}(w_u)-w_u]=\mathbb{E}[(\mathcal{M}(w_u)-w_u)^2]-(w_u-w_u)^2=\mathbb{E}[(\mathcal{M}(w_u)-w_u)^2]$,
by Bernstein's inequality, 
\begin{equation}
\begin{aligned}
&\Pr[|\overline{\mathcal{M}(w)}-\Bar{w}|\ge\lambda] \\
&= \Pr\left[\left|\sum_u(\mathcal{M}(w_u)-w_u)\right|\ge n\lambda\right]\\
&\le 2\cdot\exp\left(-\dfrac{\dfrac{1}{2}n^2\lambda^2}{\sum_u\mathbb{E}((\mathcal{M}(w_u)-w_u)^2)+ \dfrac{2n\lambda re^\epsilon}{3(e^\epsilon-1)}}\right)\\
&= 2\cdot\exp\left(-\dfrac{n^2\lambda^2}{2\sum_u\text{Var}[\mathcal{M}(w_u)]+ \dfrac{4n\lambda re^\epsilon}{3(e^\epsilon-1)}}\right)\\
&\le 2\cdot\exp\left(-\dfrac{n^2\lambda^2}{2nr^2\left(\dfrac{e^\epsilon+1}{e^\epsilon-1}\right)^2+ \dfrac{4n\lambda re^\epsilon}{3(e^\epsilon-1)}}\right)\\
&= 2\cdot\exp\left(-\dfrac{n\lambda^2}{2r^2\left(\dfrac{e^\epsilon+1}{e^\epsilon-1}\right)^2+ \dfrac{4\lambda re^\epsilon}{3(e^\epsilon-1)}}\right)\\
&= 2\cdot\exp\left(-\dfrac{n\lambda^2}{r^2O(\epsilon^{-2})+\lambda rO(\epsilon^{-1})}\right).
\end{aligned}
\end{equation}
In other words, there exists $\lambda=O\left(\dfrac{r\sqrt{-\log\beta}}{\epsilon\sqrt{n}}\right)$ such that $|\overline{\mathcal{M}(w)}-\Bar{w}|<\lambda$ holds with at least $1-\beta$ probability.
\end{proof}

\end{document}